\documentclass[11 pt]{article}
\usepackage{amsmath}
\usepackage{amsthm}
\usepackage{amssymb}
\usepackage{graphicx}
\usepackage{hyperref}
\usepackage{color}
\usepackage{mathtools}
\usepackage{graphicx}
\usepackage{float}
\newcommand{\indep}{\rotatebox[origin=c]{90}{$\models$}}

\newtheorem{theorem}{Theorem}[section]
\newtheorem{definition}[theorem]{Definition}

\newtheorem{lemma}[theorem]{Lemma}

\newtheorem{remark}[theorem]{Remark}
\numberwithin{equation}{section}

\DeclareMathOperator{\erf}{erf}
\begin{document}
\title{First exit-time analysis for an approximate Barndorff-Nielsen and Shephard model with stationary self-decomposable variance process}
\author{ Shantanu Awasthi\footnote{Email: shantanu.awasthi@ndus.edu}, Indranil SenGupta\footnote{Email: indranil.sengupta@ndsu.edu} \\ Department of Mathematics \\ North Dakota State University \\ Fargo, North Dakota, USA.}
\date{\today}
\maketitle

\begin{abstract}

In this paper, an approximate version of the Barndorff-Nielsen and Shephard model, driven by a Brownian motion and a L\'evy subordinator, is formulated. The first-exit time of the log-return process for this model is analyzed. It is shown that with a certain probability, the first-exit time process of the log-return is decomposable into the sum of the first exit time of the Brownian motion with drift, and the first exit time of a L\'evy subordinator with drift. Subsequently, the probability density functions of the first exit time of some specific L\'evy subordinators, connected to stationary, self-decomposable variance processes, are studied. Analytical expressions of the probability density function of the first-exit time of three such L\'evy subordinators are obtained in terms of various special functions. The results are implemented to empirical S\&P 500 dataset.

\end{abstract}

\textsc{Key Words:} L\'evy process, Subordinator, Self-decomposability, First-exit time, Laplace transform. 

 \section{Introduction}
\label{ch:newintro}

The time required for a stochastic process, starting at a given initial state, to reach a threshold for the first time is referred to as the first-exit time or the first hitting time. It is typically very useful in determining expected lifetime of mechanical devices. The first-exit time processes are very useful for understanding various financial sectors, especially the insurance industry and investment firms. The first-exit time processes arise naturally in the studies of various disciplines. For example, this is used in \cite{new1} to model the death probability density function for a decaying stochastic process that represents either the end of functionality for a machine, or a zero health state for an organism. The paper \cite{new2} provides an expanded first-exit time density function that expresses the human death distribution.  The first-exit time analysis of a two-dimensional symmetric stable process is discussed in detail in the paper \cite{new3}. This is further developed in \cite{pv, AK} where the first-exit time process of an inverse Gaussian L\'evy process is considered. The one-dimensional distribution function of the first-exit time process is obtained. The first-exit time analysis related to a geophysical data is provided in \cite{se1}. The paper \cite{new18}, provides generalized notions and analysis methods for the exit-time of random walks on graphs. 

The first-exit time process of the standard Brownian motion is well-studied in the literature (see \cite{Applebaum}). The paper \cite{new11} studies the first-exit time of Brownian motion for a parabolic domain. In \cite{new13}, the Fokker-Planck equation is solved for the Brownian motion with drift, in the presence of a fixed initial point and elastic boundaries. An explicit expression is obtained for the density of the first-exit time. The paper \cite{new10} studies the first-exit time problem for the solutions of some stochastic differential equations for bounded or unbounded intervals. Studies in \cite{new4, new5, new6} discuss the first-exit time process for strictly increasing L\'evy processes. In the pioneering paper \cite{new9}, the authors study the first-exit-time to flat boundaries for a double exponential jump diffusion process. The related stochastic process consists of a continuous Brownian motion-driven part, and a jump part with jump sizes given by a double exponential distribution. In general, with the help of a fluctuation identity, the paper \cite{new7} provides, a generic link between a number of known identities for the first-exit time and the overshoot above/below a fixed level of a L\'evy process. In \cite{new8}, a class of increasing L\'evy processes perturbed by an independent Brownian motion is considered, and the problem of determining the distribution of the first-exit time is addressed. The first-exit time analysis of the Ornstein-Uhlenbeck (OU) process to a boundary is a long-standing problem with no known closed-form solution for the general case. In \cite{new12} a general mean-reverting process is considered to investigate the long-and short-time asymptotics using a combination of Hopf-Cole and Laplace transform techniques. 

Many problems in finance are related to the first-exit time processes. A deeper understanding of such processes leads to a wiser estimation of fluctuations in the market. In \cite{new14}, the first-exit time distributions of stock price returns in different time windows are analyzed. The probability distribution obtained by such analysis is compared with those obtained from different models for stock market evolution. The paper \cite{new19} shows that for continuous time transformations, independent of the Brownian motion, analytical results for the double-barrier problem can be obtained via the Laplace transform of the time change. The analysis provides a power series representation for the resulting first-exit time probabilities. In \cite{new15}, explicit analytical characterizations are provided for the first-exit time densities for the Cox-Ingersoll-Ross (CIR) and OU diffusions. Such characterizations are obtained in terms of the relevant Sturm-Liouville eigenfunction expansions. In \cite{new16}, a doubly skewed CIR process is studied. A modified spectral expansion is used to obtain the first-exit time distribution of a doubly skewed CIR process. A detailed study of the first-exit times of diffusion processes and their applications to finance is provided in \cite{new17}. The studies in \cite{Roberts, Roberts1} discuss the first-exit time analysis related to some financial processes from a data-science and sequential hypothesis testing perspective. In \cite{new20},  the authors provide a solution to the optimal stopping problem of a Brownian motion subject to the constraint that the stopping time's distribution is a given measure consisting of finitely many atoms. The distribution constraints lead to an application in mathematical finance to model-free super-hedging with an outlook on volatility.

Some analytically tractable formulas are available for the density of the first-exit time process (see \cite{AK}). However, in general, an explicit expression for the density of the first-exit time process for a financial model is mostly unknown. In this paper, we analyze the first-exit time processes in connection to the Barndorff-Nielsen and Shephard (BN-S) model, a popularly used stochastic volatility model for financial analysis. In this paper we provide various analytical formulas related the distribution of the first-exit time processes in connection to an approximate version of the BN-S model. For this study we use various properties of the Laplace transform and their relations to special functions. In particular, the first-exit time processes for some well-known self-decomposable L\'evy subordinators are analyzed. 


The organization of the paper is as follows. In Section \ref{sec2}, a description of the BN-S model is provided. An \emph{approximation of the BN-S model} is formulated based on the stationary, self-decomposable distributions of the variance process. In Section \ref{sec3}, the first-exit time for a combination of a Brownian motion and a L\'evy subordinator is analyzed. In Section \ref{sec4}, the first-exit time distribution is studied in connection to various self-decomposable L\'evy subordinators. It is shown that the analysis is related to the first-exit time analysis of the log-return for the \emph{approximation of the BN-S model} presented in Section \ref{sec2}. In Section \ref{sec5} some numerical results are provided based on S\&P 500 close price dataset for a period of ten years. Finally, a brief conclusion is provided in Section \ref{sec6}.

\section{Barndorff-Nielsen and Shephard model, self-decomposability, and an approximation}
\label{sec2}

Financial time series of different assets share many common features which are successfully captured by the stochastic model introduced in various works of Ole Barndorff-Nielsen and Neil Shephard. The model is known in modern literature as the Barndorff-Nielsen and Shephard (BN-S) model (see \cite{BN1, BN-S1, BN-S2}). This model is revised and refined in various recent works in literature such as \cite{SenGupta, ijtaf}. This model is successfully implemented in the commodity markets as well (see \cite{SWW, Wil2}). Recently, this model is improved using various machine-learning driven algorithms (see \cite{Wilson, Humay}).

For the BN-S model, a frictionless financial market is considered where a risk-less asset with constant interest rate $r$, and a stock, are traded up to a fixed horizon date $T$. It is assumed that the price process of the stock $S= \{S_t\}_{t \geq 0}$ is defined on some filtered probability space $(\Omega, \mathcal{F}, (\mathcal{F}_t)_{0 \leq t \leq T}, \mathbb{P})$ and is given by:

\begin{equation}
\label{1}
S_t= S_0 \exp (X_t),
\end{equation}
where the log-return $X_t$ is given by
\begin{equation}
\label{se411}
dX_t = (\mu_1 + \beta_1 \sigma_t ^2 )\,dt + \sigma_t\, dW_t + \rho \,dZ_{\lambda t}, 
\end{equation}
with the variance process
\begin{equation}
\label{se421}
d\sigma_t ^2 = -\lambda \sigma_t^2 \,dt + dZ_{\lambda t}, \quad \sigma_0^2 >0,
\end{equation}
where the parameters $\mu_1, \beta_1 \in \mathbb{R}$ with $\lambda >0$ and $\rho < 0$. In \eqref{se411} and \eqref{se421}, $W_t$ and $Z_t$ are a Brownian motion and a L\'evy subordinator, respectively. The L\'evy subordinator $Z$ is referred to as the background driving L\'evy process (BDLP). Also $W$ and $Z$ are assumed to be independent and $(\mathcal{F}_t)$ is assumed to be the usual augmentation of the filtration generated by the pair $(W, Z)$. Without loss of generality, we assume $W_0=Z_0=0$.

We assume $Z$ satisfies the assumptions described in \cite{NV}. It follows that the cumulant transform $\kappa(\theta)= \log E[e^{\theta Z_1}]$, where it exists, takes the form
$\kappa(\theta)= \int_{\mathbb{R}_{+}} (e^{\theta x} -1) w(x) \, dx$, where $w(x)$ is the L\'evy density for $Z$. It is shown in \cite{NV} (Theorem 3.2) that there exists an equivalent martingale measure (EMM) $\mathbb{Q}$, under which equations \eqref{se411} and \eqref{se421} can be written as:
\begin{equation}
\label{se41}
dX_t=b_t \,dt+ \sigma_t\, dW_t + \rho\, dZ_{\lambda t}, \quad \text{with} \quad b_t=  (r- \lambda \kappa (\rho)- \frac{1}{2} \sigma_t^2),
\end{equation}
\begin{equation}
\label{se42}
d\sigma_t^2= - \lambda \sigma_t^2\, dt + dZ_{\lambda t}, \quad \sigma_0^2>0,
\end{equation}
where  $W_t$ and $Z_{t}$  are a Brownian motion and a L\'evy subordinator respectively with respect to $\mathbb{Q}$. For the rest of this paper we assume that the risk-neutral dynamics (with respect to $\mathbb{Q}$) of the stock price is given by \eqref{1}, \eqref{se41} and \eqref{se42}.
It is trivial to show that the solution of \eqref{se42} is given by
\begin{equation}
\label{se43}
\sigma_t^2= e^{-\lambda t}\sigma_0^2 + \int_0^t e^{-\lambda (t-s)}\, dZ_{\lambda s}.
\end{equation}
From \eqref{se43}, the positivity of the process $\sigma_t^2$ is obvious. In fact, $\sigma_t^2$ is bounded below by the deterministic function $e^{-\lambda t}\sigma_0^2$. In addition, the instantaneous variance of log-return $X_t$ is given by $(\sigma_t^2+  \rho^2 \lambda \text{Var}[Z_1])\,dt$. Consequently, the continuous realized variance in the interval $[0,T]$, denoted as $\sigma_R^2$, is given by $\sigma_R^2=\frac{1}{T}\int_0^T \sigma_t^2\,dt + \rho^2 \lambda \text{Var}[Z_1]$.
Therefore, by \eqref{se43} we obtain
\begin{equation}
\label{se44}
\sigma_R^2=\frac{1}{T} \left(\lambda^{-1}(1-e^{-\lambda T})\sigma_0^2+\lambda^{-1}\int_0^T\left(1-e^{-\lambda(T-s)}\right)dZ_{\lambda s} \right)+ \rho^2 \lambda \text{Var}[Z_1].
\end{equation}


We state some results for the analysis of the variance process $\sigma_t^2$, when the process is stationary and self-decomposable. The results are motivated by \cite{Semere, Semere2, Issaka}. The pricing formulas for various derivatives are dependent on the variance process. 
\begin{definition}
The distribution of a random variable $X$ is said to be self-decomposable if for any constant $c$, $0 < c < 1$, there exists an independent random variable $X^{(c)}$, such that $X \,{\buildrel \text{d} \over =}\, cX+ X^{(c)}$, where $\,{\buildrel \text{d} \over =}\,$ stands for the equality in the distribution. 
\end{definition}
For self-decomposable laws the associated densities are unimodal (see \cite{Carr, Sato}). It is proved in  \cite{BJS, Wolfe} that, if $X$ is self-decomposable then there exists a \emph{stationary stochastic process} $\{\sigma^2(t)\}_{t \geq 0}$, and a L\'evy subordinator $\{Z_t\}_{t \geq 0}$, independent of $\sigma_0^2$, such that $\sigma_t^2 \,{\buildrel d \over =}\, X$ for all $t \geq 0$ and 
\begin{equation*}
\sigma_t^2= \exp(-\lambda t) \sigma_0^2+ \int_0^t \exp\left(-\lambda (t-s)\right) \, dZ_{\lambda s}, \quad \text{for all   } \lambda>0.
\end{equation*}
Conversely, if $\{ \sigma_t^2\}_{t \geq 0}$, is a \emph{stationary stochastic process} and $\{Z_t\}_{t \geq 0}$ is a L\'evy subordinator independent of $\sigma_0^2$, such that $\{\sigma_t^2\}$ and $\{Z_t\}$ satisfy 
$$d\sigma_t^2= - \lambda \sigma_t^2\, dt+ dZ_{\lambda t}, \quad \sigma_0^2>0,$$ for all $\lambda>0$, then $\sigma_t^2$ is self-decomposable. 

It is clear from \cite{Sato} (Theorem 17.5(ii)) that for any self-decomposable law $D$ there exists a L\'evy subordinator $Z$ such that the process of OU type driven by $Z$ has invariant distribution given by $D$. The following theorem (see \cite{Semere, Semere2, ijtaf}) gives the relation between the L\'evy densities of such process generated by $\sigma_t^2$ and $Z$ in \eqref{se42}. 
\begin{theorem}
A random variable $X$ has law in $L$ if and only if $X$ has a representation of the form
$X= \int_0^{\infty} e^{-t}\, dZ_t$, where $Z_t$ is a L\'evy subordinator. In this case the L\'evy measure $U$ and $W$ of $X$ and $Z_1$ are related by $U(dx)= \int_0^{\infty} W(e^t\, dx)\,dt$.
In addition, if $u(x)$, the L\'evy density of $U$ is differentiable, then the L\'evy measure $W$ has a density $w$, and $u$ and $w$ are related by
\begin{equation}
\label{62}
w(x)=-u(x)-xu'(x).
\end{equation}
\end{theorem}

There are many known self-decomposable distributions, such as inverse Gaussian (IG), Gamma, positive tempered stable (PTS), etc. 

Consequently, if the stationary distribution of $\sigma^2_t$ is given by $\text{IG}(\delta_1, \gamma)$ law, with the L\'evy density  $u(x)= \frac{1}{\sqrt{2\pi}} \delta_1 x^{-3/2}\exp(-\gamma^2x/2)$,  $x>0$, then by \eqref{62}, the L\'evy density of $Z_1$ is given by $w(x)= \frac{\delta_1}{2\sqrt{2 \pi}} x^{-\frac{3}{2}} (1+ \gamma^2 x)e^{-\frac{1}{2} \gamma^2 x}$, $x>0$. Alternatively, if the stationary distribution of $\sigma_t^2$ is given by gamma law $\Gamma(\nu, \alpha)$, where the L\'evy density of $\Gamma(\nu,\alpha)$ is given by $u(x)=\nu x^{-1}e^{-\alpha x}$, $x>0$, then by \eqref{62} we obtain $w(x)= \nu \alpha e^{-\alpha x}$, $x>0$. 

A three-parameter self-decomposable process is positive tempered stable (PTS) process (see \cite{boyar1, boyar2}). It is denoted as $\text{PTS}(\kappa, \delta, \gamma)$, where $\beta>0$, $0<\gamma<1$, and  $k \geq 0$. For $\text{PTS}(\kappa, \delta, \gamma)$ process the L\'evy density is simple and is given by (see \cite{Semere, Semere2}) $$u(x)= \beta k^{-2 \gamma} \frac{\gamma}{\Gamma(\gamma) \Gamma(1-\gamma)}x^{-\gamma-1}\exp\left(-\frac{1}{2}k^2 x\right), \quad x>0.$$
If the stationary distribution of $\sigma^2_t$ is given by $\text{PTS}(\kappa, \delta, \gamma)$ law,  then by \eqref{62} we obtain that the L\'evy density of $Z_1$ is given by  
\begin{equation*}
w(x) = \frac{\beta k^{-2\gamma} \gamma x^{-\gamma-1} e^{\frac{-k^2 x}{2}}}{\Gamma(\gamma) \Gamma(1-\gamma)}\left(\gamma + \frac{k^2 x}{2}\right), \quad x>0.
\end{equation*}

In the above discussions we find that the distribution of $Z$ is analytically tractable when the stationary distribution of $\sigma_t^2$ in \eqref{se42} is given by a stationary, self-decomposable distribution. We denote (as $\sigma_t^2$ is stationary), 
\begin{equation}
\label{sigma1}
\sigma= E^{\mathbb{Q}}(\sigma_1^2),
\end{equation}
 and 
\begin{equation}
\label{mu1}
\mu = r- \lambda \kappa (\rho)- \frac{1}{2} \sigma^2.
\end{equation}
We approximate \eqref{se41} by
\begin{equation}
\label{shan41}
dX_t= \mu \,dt+ \sigma\, dW_t + \rho\, dZ_{\lambda t}.
\end{equation}
We refer to  \eqref{1}, \eqref{shan41}, and \eqref{se42}, as an \emph{approximation of the BN-S model} \eqref{1}, \eqref{se41}, and \eqref{se42}. For most of the empirical financial data $\mu \leq 0$.

We write $X_t= \mu t+ \sigma W_t+ \rho Z_t$, with $\mu \in \mathbb{R}$, $\sigma>0$, and $\rho<0$, $t>0$. For financial applications $\mu \leq 0$. For the subsequent sections we develop a general procedure to compute the first-exit time of the stochastic process $X_t$.

\section{First-exit time for a combination of a Brownian motion and a L\'evy subordinator}
\label{sec3}

In this section, we develop a couple of results related to the first-exit time analysis of log-return processes of the form \eqref{shan41}. At first, we develop the result related to the first-exit time of a simpler process $W_t +Y_t$, where $Y$ is a L\'evy subordinator, with $W_0=Y_0=0$. If $X_1$ and $X_2$ are independent random variables, we denote $X_1 \indep X_2$.


\begin{theorem}
\label{1111}
For a Brownian motion $W_t$ and a L\'evy subordinator $Y_t$, and $a, b >0$, 
\begin{equation}
\label{inf}
\inf \{\tau > 0: W_{\tau} + Y_{\tau} \geq a + b \}= \inf\{t > 0: W_t \geq a \}+ \inf \{\alpha > 0: Y_{\alpha} \geq b \},
\end{equation}
with probability 
\begin{equation}
\label{prob111}
P=\int_0^{\infty}\int_0^{\infty}\int_{-\infty}^{\infty} P_1(\epsilon; t, \alpha)P_2(\epsilon; t, \alpha) \,d\epsilon\, dt\, d\alpha,
\end{equation}
where 
\begin{equation}
\label{p1}
P_1(\epsilon; t, \alpha) = \int_{-\infty}^{\infty} \frac{e^\frac{-\tau^2}{2\alpha}}{\sqrt{2\pi\alpha}} \left(\int_{\max(a,a-\epsilon-\tau)}^{\infty} \frac{e^\frac{-s^2}{2t}}{\sqrt{2\pi t} }ds\right)\, d\tau,
\end{equation}
and 
\begin{equation}
\label{p2}
P_2(\epsilon; t, \alpha) = \int_{0}^{\infty} f_{Y_t}(\beta) \left(\int_{\max(\max(b,b+\epsilon-\beta),0)}^{\infty} f_{Y_{\alpha}}(s)  ds\right) d\beta,
\end{equation}
where the probability density function of $Y_t$ is given by $f_{Y_t}(\cdot)$. 
\end{theorem}

\begin{proof}
The first-exit time of a combination of $W_t$ and $Y_t$, in the sense that its value is more than $a+b$, is given by 
\begin{align*}
\inf \{\tau> 0: W_{\tau} + Y_{\tau} \geq a + b \} = \inf \{t + \alpha > 0: W_{t+\alpha} + Y_{t+\alpha} \geq a + b, t>0, \alpha>0 \}.
\end{align*}
For a fixed $\epsilon\in \mathbb{R}$, we define 
\begin{equation}
P_1(\epsilon; t, \alpha) =  P(W_{t+\alpha} \geq a - \epsilon, W_t \geq a),
\end{equation}
\begin{equation}
P_2(\epsilon; t, \alpha)= P(Y_{t+\alpha} \geq b + \epsilon, Y_{\alpha}\geq b).
\end{equation}
We proceed to compute $P_1(\epsilon; t, \alpha)$ and $P_2(\epsilon; t, \alpha)$. We observe,
\begin{align*}
P_1(\epsilon; t, \alpha) &= P(W_{t+\alpha} \geq a - \epsilon ,  W_t \geq a) \nonumber \\
& = P(W_{t + \alpha} - W_t \geq a - \epsilon -W_t, W_t \geq a)  \nonumber \\
& =  P( W_t \geq a - \epsilon -(W_{t + \alpha} - W_t), W_t \geq a )  \nonumber \\
& = P\left(W_t \geq \max(a,a-\epsilon-\chi)\right), \quad \chi \sim \mathcal{N}(0, \alpha), \quad \chi \indep W_t \nonumber \\
& = \int_{-\infty}^{\infty} \frac{e^\frac{-\tau^2}{2\alpha}}{\sqrt{2\pi\alpha}} \left(\int_{\max(a,a-\epsilon-\tau)}^{\infty} \frac{e^\frac{-s^2}{2t}}{\sqrt{2\pi t} }\,ds\right) \,d\tau.
\end{align*}
On the other hand, 
\begin{align*}
P_2(\epsilon; t, \alpha) & =   P(Y_{t+\alpha} \geq b + \epsilon , Y_{\alpha} \geq b)  \\
& = P(Y_{t + \alpha} - Y_{\alpha} \geq b + \epsilon -Y_{\alpha}, Y_{\alpha} \geq b) \\
& =  P( Y_{\alpha} \geq b + \epsilon -(Y_{t + \alpha} - Y_{\alpha}), Y_{\alpha} \geq b) \\
& = P(Y_{\alpha} \geq \max(b,b+\epsilon-\eta)), \quad \eta \sim \text{the distribution of } Y_{t}, \quad \eta \indep Y_{\alpha}.
\end{align*}
As the probability density function of $Y_t$ is given by $f_{Y_t}(\cdot)$, therefore we obtain 
\begin{equation*}
  P_2(\epsilon; t, \alpha)=  \int_{0}^{\infty} f_{Y_t}(\beta) \left(\int_{\max(\max(b,b+\epsilon-\beta),0)}^{\infty} f_{Y_{\alpha}}(s)  ds\right) d\beta.
\end{equation*}
Clearly, $\{t> 0: W_t \geq a \}+ \{\alpha > 0: Y_{\alpha} \geq b \}= \{t + \alpha > 0: W_{t+\alpha} + Y_{t+\alpha} \geq a + b, t>0, \alpha>0 \}$, with probability $P$, where $P$ is given by \eqref{prob111}, and $ P_1(\epsilon; t, \alpha)$ and $ P_2(\epsilon; t, \alpha)$ are obtained by \eqref{p1} and \eqref{p2}, respectively. This leads to \eqref{inf}.
\end{proof}

Next, we generalize the result in Theorem \ref{1111} for the log-return stochastic process \eqref{shan41} in the \emph{approximation of the BN-S model}. In the BN-S model $\rho<0$ is assumed in order to incorporate the \emph{leverage effect} of the market. Typically in a derivative market, a significant fluctuation always corresponds to a ``big-downward-movement" of the asset prices. Consequently, for the next theorem we focus on the first-exit time corresponding to a ``downward-movement" of the log-return process \eqref{shan41}. For the following theorem we assume $W_0=Z_0=0$.

\begin{theorem}
\label{maintheorem}
For a Brownian motion $W_t$ and a L\'evy subordinator $Z_t$, if $\mu \in \mathbb{R}$, $\sigma>0$, $\rho <0$, and $a,b>0$, then
\begin{align}
\label{infgen}
& \inf \{\tau > 0 : \mu \tau + \sigma W_{\tau} + \rho Z_{\tau} \leq  -a- b\}  \nonumber \\
& = \inf \{t_1>0: \mu t_1 +\sigma W_{t_1} \leq -a \} + \inf \{t_2>0: \mu t_2 + \rho Z_{t_2} \leq -b \},
\end{align}
with probability 
\begin{equation}
\label{probgen}
P=\int_0^{\infty}\int_0^{\infty} \left( \int_{-\infty}^{\infty} P_1(\epsilon; t_1,t_2)P_2(\epsilon; t_1,t_2) \,d\epsilon \right)\, dt_1\, dt_2,
\end{equation}
where 
\begin{equation}
\label{p1gen}
P_1(\epsilon; t_1,t_2) = \int_{-\infty}^{\infty} \frac{e^\frac{-\tau^2}{2 t_2}}{\sqrt{2\pi t_2}} \left(\int_{-\infty}^{\min\left( \frac{(-a-\mu t_1)}{\sigma}, \frac{(-a-\epsilon)}{\sigma}-\tau -\frac{\mu(t_1+ t_2)}{2\sigma}\right)} \frac{e^\frac{-s^2}{2t_1}}{\sqrt{2\pi t_1} }ds\right) d\tau,
\end{equation}
and 
\begin{equation}
\label{p2gen}
P_2(\epsilon; t_1, t_2) = \int_{0}^{\infty} f_{Z_{t_1} }(\beta) \left(\int_{\max\left(\max\left(\frac{(b-\mu t_2)}{\rho}, \frac{(b+\epsilon)}{\rho}-\beta-\frac{\mu(t_2+t_1)}{2\rho}\right), 0\right)}^{\infty} f_{Z_{t_2} }(s)  ds\right) d\beta,
\end{equation}
where the probability density function of $Z_t$ is given by $f_{Z_t}(\cdot)$. 
\end{theorem}

\begin{proof}
For fixed $\epsilon \in \mathbb{R}$, we define and compute the following joint probabilities. At first, we compute, for $a>0$:
\begin{align*}
P_1(\epsilon; t_1,t_2) &= P(W_{t_1+ t_2} +\frac{\mu(t_1 + t_2)}{2\sigma} \leq \frac{-a}{\sigma} - \frac{\epsilon}{\sigma}, W_{t_1} +\frac{\mu t_1}{\sigma} \leq \frac{-a}{\sigma}) \nonumber \\
& = P(W_{t_1 + t_2} - W_{t_1} + \frac{\mu t_1}{2\sigma} \leq \frac{-a}{\sigma} - \frac{\epsilon}{\sigma} -W_{t_1}- \frac{\mu t_2}{2\sigma}, W_{t_1}+\frac{\mu t_1}{\sigma} \leq \frac{-a}{\sigma})  \nonumber \\
& =  P( W_{t_1} +\frac{\mu t_1}{2\sigma} \leq \frac{-a}{\sigma}- \frac{\epsilon}{\sigma} -(W_{t_1 + t_2} - W_{t_1}) -\frac{\mu t_2}{2\sigma}, W_{t_1} +\frac{\mu t_1}{\sigma}\leq \frac{-a}{\sigma})  \nonumber \\
& =  P\left( W_{t_1}  \leq \frac{(-a- \epsilon)}{\sigma} -(W_{t_1 + t_2} - W_{t_1}) -\frac{\mu t_2}{2\sigma}-\frac{\mu t_1}{2\sigma}, W_{t_1} \leq \frac{-a}{\sigma} -\frac{\mu t_1}{\sigma}\right) \nonumber \\
& = P\left(W_{t_1}  \leq \min\left( \frac{(-a-\mu t_1)}{\sigma}, \frac{(-a-\epsilon)}{\sigma}-\chi -\frac{\mu(t_1+ t_2)}{2\sigma}\right)\right), \quad \chi \sim \mathcal{N}(0, t_2), \quad \chi \indep W_{t_!} \nonumber \\
& = \int_{-\infty}^{\infty} \frac{e^\frac{-\tau^2}{2 t_2}}{\sqrt{2\pi t_2}} \left(\int_{-\infty}^{\min\left( \frac{(-a-\mu t_1)}{\sigma}, \frac{(-a-\epsilon)}{\sigma}-\tau -\frac{\mu(t_1+ t_2)}{2\sigma}\right)} \frac{e^\frac{-s^2}{2t_1}}{\sqrt{2\pi t_1} }ds\right) d\tau.
\end{align*}
With $\rho<0$, we compute for $b>0$,
\begin{align*}
P_2(\epsilon; t_1, t_2) & =   P(Z_{t_1+t_2} + \frac{\mu(t_1 + t_2)}{2\rho}\geq \frac{-b}{\rho} + \frac{\epsilon}{\rho}, Z_{t_2} +\frac{\mu t_2}{\rho}\geq \frac{-b}{\rho})  \\
& = P(Z_{t_1 + t_2} + \frac{\mu(t_1 + t_2)}{2\rho} -Z_{t_2} \geq \frac{-b}{\rho}  + \frac{\epsilon}{\rho} -Z_{t_2} , Z_{t_2} +\frac{\mu t_2}{\rho} \geq \frac{-b}{\rho} ) \\
& =  P\left( Z_{t_2}   \geq \frac{(-b+\epsilon)}{\rho} -((Z_{t_1+t_2}-Z_{t_2})+\frac{\mu(t_2 + t_1)}{2\rho}), Z_{t_2} \geq \frac{-b}{\rho}  -\frac{\mu t_2}{\rho}\right) \\
& = P\left(Z_{t_2} \geq  \max\left(\frac{(-b-\mu t_2)}{\rho}, \frac{(-b+\epsilon)}{\rho}-\eta-\frac{\mu(t_2+t_1)}{2\rho}\right)\right),
\end{align*}
where $\eta \sim \text{the distribution of }Z_{t_1}$. Since $\eta \indep Z_{t_2}$, therefore we obtain \eqref{p2gen}. For $a,b>0$, we define a set 
\begin{align*}
A & = \{\tau > 0 : \mu \tau + \sigma W_{\tau} + \rho Z_{\tau} \leq  -a- b\} \\
&= \{t_1 + t_2 > 0 : \mu (t_1+t_2) + \sigma W_{t_1 + t_2} + \rho Z_{t_1+t_2} \leq  -a- b, t_1>0, t_2>0\}.
\end{align*}
Consequently, we obtain 
\begin{align*}
A&= \{t_1 + t_2 > 0 : \mu (t_1+t_2) + \sigma W_{t_1 + t_2} + \rho Z_{t_1+t_2} \leq  -a- b\}  \nonumber \\
& = \{t_1>0: \mu t_1 +\sigma W_{t_1} \leq -a \} +  \{t_2>0: \mu t_2 + \rho Z_{t_2} \leq -b \},
\end{align*}
with probability $P$ given by \eqref{probgen}. Consequently, $$\inf A= \inf \{t_1>0: \mu t_1 +\sigma W_{t_1} \leq -a \} + \inf \{t_2>0: \mu t_2 + \rho Z_{t_2} \leq -b \}.$$
This proves \eqref{infgen}. 
\end{proof}

The purpose of Theorem \ref{1111} and Theorem \ref{maintheorem} is to decompose the first-exit time process of a linear combination of a Brownian motion and a L\'evy subordinator into the individual first-exit time processes of a Brownian motion and a L\'evy subordinator. However, as observed in both of the theorems, such decomposition holds only with certain probability.

\begin{remark}
\label{remar}
It is well known (see \cite{Applebaum, pv})  that for the process $G_t= \inf \{s>0:W_s+ \gamma s \geq \delta_1 t\}$, with $\gamma, \delta_1>0$, known as the inverse Gaussian (IG) process, $G_t$ follows an $\text{IG}(\delta_1 t, \gamma)$ distribution. As the process $W_s+ \gamma s$ is continuous, we also have $G_t= \inf \{s>0:W_s+ \gamma s = \delta_1 t\}$. The distribution $\text{IG}(\delta_1, \gamma)$ is concentrated on $\mathbb{R}_+$ and has probability density:
$$p(x)= \frac{1}{\sqrt{2\pi}}\delta_1 e^{\delta_1 \gamma} x^{-3/2} \exp\left(-\frac{\delta_1^2 x^{-1} + \gamma^2 x}{2} \right), \quad \gamma, \delta_1 > 0.$$ 
Consequently, for Theorem \ref{maintheorem} with $\mu<0$ and $a>0$, the first term on the right hand side of \eqref{infgen} has the distribution $  \inf \{t_1>0: \mu t_1 +\sigma W_{t_1} \leq -a \} \,{\buildrel d \over =}\,  \inf \{t_1>0: -\mu t_1 +\sigma W_{t_1} \geq a \} \sim \text{IG}\left(\frac{a}{\sigma}, \frac{-\mu}{\sigma}\right)$.

The case is not the same if the Brownian motion does not have any drift term. In that case, it is known (see \cite{Bac}) that $\inf\{s>0: \sigma W_s \geq a \}$, with $a>0$, satisfies a L\'evy distribution with the probability density function $$\frac{a}{\sigma \sqrt{2\pi x^3}}\exp\left(-\frac{a^2}{2\sigma^2x}\right), \quad x>0.$$ Consequently, for Theorem \ref{1111}, the first term on the right hand side of \eqref{inf}, i.e., $\inf\{t > 0: W_t \geq a \}=\inf\{t > 0: W_t = a \}$, with $a>0$, has the probability density function $\frac{a}{ \sqrt{2\pi x^3}}\exp\left(-\frac{a^2}{2x}\right)$, $x>0$. A similar result holds for the first term on the right hand side of \eqref{infgen} in Theorem \ref{maintheorem} with $\mu=0$.

Note that, for the case when $\mu=0$ and $a>0$, $\inf\{s>0: \sigma W_s \leq -a \}= \inf\{s>0: \sigma W_s = -a \}\,{\buildrel d \over =}\,\inf\{s>0: \sigma W_s = a \}= \inf\{s>0: \sigma W_s \geq a \}$.

\end{remark}

We note that for Theorem \ref{1111}, if $a,b\leq 0$, then \eqref{inf} is trivially satisfied. Similarly, for Theorem \ref{maintheorem}, if $a, b \leq 0$, then \eqref{infgen} is trivially satisfied. As $W_0=Z_0=0$, therefore all the related first-exit times are zero in those cases.


\section{First-exit time distribution for some self-decomposable processes}
\label{sec4}

Consider the log-return dynamics $X_t$ given by \eqref{shan41}, in the \emph{approximation of the BN-S model} \eqref{1}, \eqref{shan41}, and \eqref{se42}. In Theorem \ref{maintheorem}, it is shown that with certain probability, the first-exit time process $\inf \{t > 0 : X_t \leq  -a- b\}$, is decomposable into the sum of the first exit time of two processes- (1) the Brownian motion with drift, and (2) a L\'evy subordinator with drift. We denote three stochastic processes: $A_{a+b}=\inf \{t > 0 : X_t \leq  -a- b\}=\inf \{t > 0 : \mu t+ \sigma W_{t} + \rho Z_{t} \leq  -a- b\}$, $B_a= \inf \{t>0: \mu t +\sigma W_{t} \leq -a \}$, and $C_b= \inf \{t>0: \mu t+ \rho Z_{t} \leq -b\}$, with $\rho<0$, and $a,b>0$. In these expressions $\sigma$ and $\mu$ are given by \eqref{sigma1} and \eqref{mu1}, respectively. Thus, $\sigma>0$. Also, in general, for financial applications $\mu \leq 0$. With these notations, from Theorem \ref{maintheorem} we obtain that $A_{a+b}= B_a+C_b$. 

The probability density function of the process $B$, with $\mu \leq 0$, is discussed in Remark \ref{remar}. In this section we discuss the probability density function of the process $C$ for some special cases. Accordingly, with probability $P$ given by \eqref{probgen}, the probability density function of the process $A$ is equal to the convolution of the probability density functions of the processes $B$ and $C$.

The goal of this section is to analyze the first-exit time distribution for the L\'evy subordinator in the decompositions provided in Theorem \ref{1111} and Theorem \ref{maintheorem}. For simplicity we assume $\mu=0$. We consider the distribution of the corresponding process $C_b=\inf \{s>0: Z_{s} \geq \frac{-b}{\rho} \}$, for three self-decomposable distributions. As $b>0$ and $\rho<0$, in general, $C$ can be written as the stochastic process $T_t= \inf \{s>0: Z_{s} \geq t \}$, $t>0$.

In Subsection \ref{sec41}, we describe some results related to special functions and Laplace transforms that are implemented for the subsequent analysis. Subsections \ref{sec42}, \ref{sec43}, and \ref{sec44}, deal with various analysis of $T_t$ in relation to Gamma, IG, and PTS subordinators, respectively.

\subsection{Laplace transform and some relevant special functions}
\label{sec41}

At first, we describe some special functions necessary for the development of the rest of this paper.
\begin{itemize}
\item The MacRobert $E$-function is denoted as
\begin{equation*}
E(m;a_1:n;b_j:x) = E(a_1,\cdots,a_m:b_1,\cdots,b_n:x).
\end{equation*}
For $m \geq n+1$, with $|x|< 1$,  the MacRobert $E$-function is defined as
\begin{equation*}
\sum_{i=1}^{m}\frac{\prod_{j=1}^{m} \tilde{*} \Gamma(a_j - a_i)\Gamma(a_i)x^{a_i}}{\prod_{k=1}^{n} \Gamma(b_k - a_i)}{_{n+1}}\tilde{F}{_{m-1}}\left[\begin{matrix}a_i,a_i-b_1+1,\cdots,a_i-b_n+1;\\a_i-a_1+1,\cdots,\tilde{*},\cdots,a_i-a_m+1;\end{matrix}(-1)^{m+n}x\right].
\end{equation*}
For $m \leq n+1$, with $|x|>1$, the MacRobert $E$-function is defined as
\begin{equation*}
\frac{\prod_{i=1}^{m}\Gamma(a_i)}{\prod_{j=1}^{n}\Gamma(b_j)} {_m}\tilde{F}{_n}\left[\begin{matrix}a_1,\cdots,a_m;\\b_1, \cdots,b_n;\end{matrix}\frac{-1}{x}\right].
\end{equation*}
For $n=0$, the notation $E(\cdot::\cdot)$ is used. The $\tilde{*}$ denotes that the term containing $a_j - a_i$ corresponding to $j=i$ is omitted. Here ${_m}\tilde{F}{_n}[\cdot]$ is generalized hypergeometric functions, defined as
\begin{equation*}
\begin{split}
{_m}\tilde{F}{_n}\left[\begin{matrix}a_1,\cdots,a_m;\\b_1, \cdots,b_n;\end{matrix}x\right] &= \sum_{n=0}^{\infty} \frac{(a_1)_1\cdots (a_m)_n x^n}{(b_1)_1 \cdots (b_n)_n n!},
\end{split}
\end{equation*}
where $(\cdot)_n$ is the Pochhammer symbol. 
\item The Gauss hypergeometric function $_2F_1\left(a,b,c;x\right)$ is defined as 
\begin{equation*}
_2F_1\left(a,b,c;x\right)= \sum_{n=0}^{\infty} \frac{(a)_n (b)_n x^n}{(c)_n n!},
\end{equation*}
where $(\cdot)_n$ is the Pochhammer symbol, $c \neq 0,-1,-2,\dots; $, and $ |x| \le 1 $. For $x \in \mathbb{C}$, with $|x| \geq 1$, the series can be analytically continued along any path in the complex plane that avoids the branch points 1 and infinity.
An integral representation of the hypergeometric function is given by
$_2F_1\left(a,b,c;x\right)= \frac{\Gamma(c) \int_{0}^{1} t^{\ b-1}  (1-t)^{\ c-b-1}(1-xt)^{\ -a} dt}{\Gamma(b) \Gamma(c-b)}$.
\item Modified Bessel functions are solutions of the modified Bessel equation. The modified Bessel function of the first kind is defined by $I_{\nu}(z)= i^{-\nu}J_{\nu}(iz)$, with $\nu \in \mathbb{R}$, and $J_{\nu}(\cdot)$ is the Bessel function of the first kind.  
\item Upper incomplete gamma function is given by
\begin{equation*}
\Gamma(a,x) = \int_{x}^{\infty}t^{a-1}e^{-t} dt.  
\end{equation*}
For $x > 0$, $\Gamma(a,x)$ converges for all real $a$. In particular, $\Gamma (0, x)$ is the \emph{exponential integral} $\int_x^{\infty} t^{-1}e^{-t}\,dt$.   
\end{itemize}

Next, we describe some results related to the Laplace transform. For $t \geq 0$, and $s \in \mathbb{C}$, we denote the Laplace transform of $f(t)$ by $\mathcal{L}(f(t))= F(s)$, where $f(t$) is piecewise continuous function on every finite interval in $[0 ,\infty)$ satisfying $|f(t)| < Me^{at}$, for some $M>0$ and for all $ t \in$ $[0 ,\infty)$. The Laplace transform and the inverse Laplace transform are related by: 

$$F(s)= \int_0^{\infty} f(t) e^{-st}\,dt,$$ and

$$f(t) = \frac{1}{2\pi i}\int_{x_0 -i\infty}^{x_0 +i \infty} e^{st}\ F(s) ds,$$ for some $x_0 \in \mathbb{R}$, where $x_0$ is greater than the real part of all singularities of $F(s)$, and $F(s)$ is bounded on the line $\text{Re}(s)= x_0$ in the complex-plane. We list some useful properties related to the Laplace transform. The following result is elementary and can be found in \cite{HK}.
\begin{lemma}
\label{lem1}
The following results hold: (1) $\mathcal{L}^{-1}\left( aF(as-b)\right)= e^{\frac{bt}{a}}f(\frac{t}{a})$, with $a>0$, $b \in \mathbb{R}$; (2) $\mathcal{L}^{-1}\left(-\frac{d F(s)}{ds}\right) =tf(t)$; (3) $\mathcal{L}^{-1}\left(\frac{F(s)}{s}\right) = \int_{0}^{t} f(u) du$; (4) $\mathcal{L}^{-1}\left(sF(s)-f(0)\right)=\frac{d f(t)}{dt}$.
\end{lemma}
The following results provide various relations between the Laplace transform and special functions. These results can be found in \cite{HK}. 
\begin{lemma}
\label{2sp}
The following results hold.
\begin{itemize}
\item[(1)] $\mathcal{L}\left(t^{\frac{-3}{2}} \int_{0}^{\infty} ue^{\frac{-u^2}{4t}}f(u) du\right)= 2(\sqrt{\pi})F(\sqrt{s})$.
\item[(2)] $\mathcal{L}^{-1}\left(\frac{e^{\frac{a}{s}}}{s}\right) = I_0 (2\sqrt{at})$, where $I_0(x)$ is the modified Bessel function of the first kind, and $\text{Re}(s) > 0$.
\item[(3)] $\mathcal{L}^{-1}\left(e^{-a\sqrt{s}}\right) =\frac{ae^{\frac{-a^2}{4t}}}{2\sqrt{\pi}t^{\frac{3}{2}}}$,  $\text{Re}(a^2) > 0, $ $\text{Re}(s) > 0$.
\item[(4)] $\mathcal{L}\left(\Gamma(v,at)\right) = \frac{\Gamma(v)}{s}[1-(1+\frac{s}{a})^{-v}]$, where $\Gamma(v,at)$ is the upper incomplete gamma function, and  $\text{Re}(\nu) > 0, \quad  \text{Re}(s) > -\text{Re}(a)$. 
\item[(5)] $\mathcal{L}\left(\erf(\sqrt{at})\right) = \frac{\sqrt{a}}{s\sqrt{s+a}}$, where $\erf (x)= \frac{2}{\sqrt{\pi}}\int_0^x e^{-t^2}\,dt$, $\text{Re}(s) > \max(0,-\text{Re}(a))$.
\item[(6)] $\mathcal{L}^{-1}\left(\frac{1}{\sqrt{s+a}}\right) = \frac{e^{-at}}{\sqrt{\pi t}}$,  $\text{Re}(s) > -\text{Re}(a)$.
 \item[(7)] $\mathcal{L}^{-1}\left(\frac{\sqrt{s+a}}{s}\right) = \frac{e^{-at}}{\sqrt{\pi t}} + \sqrt{a}\erf[\sqrt{at}] $,  $\text{Re}(s) > \max(0,-\text{Re}(a))$.
\item[(8)]  $\mathcal{L}^{-1}\left(s^{c-1}e^{-(bs)^{\frac{1}{m}}}\right)=\frac{m^{\frac{1}{2} + mc}}{(2\pi)^{\frac{m+1}{2}} b^{c}} \sum_{i,-i} \frac{1}{i} E\left(c,c+\frac{1}{m},\dots,c+ \frac{m-1}{m}:: \frac{be^{i\pi}}{m^{m}t}\right)$, where $E(\cdot:\cdot:\cdot)$ is the MacRobert $E$-function, $\text{Re}(s) > 0, \text{Re}(c) > 0,\text{Re}(b) > 0, m=2,3, \dots$.
In the above expression $\sum_{i,-i}$ denotes that in expression following the summation sign, $i$ is to be replaced by $-i$ and two expressions are to be added. 
\end{itemize}
\end{lemma}

In the two-dimension, for $x \in \mathbb{R}$, let $ F(x,s) = \int_{0}^{\infty} f(x,t) e^{\ -st}  dt$, be the Laplace transform of function $f(x, t)$ with respect to the $t$ variable. Note that, for a subordinator $X_t$, with probability density function $f_{X_t}(\cdot)$, and L\'evy measure $\pi_X$, the L\'evy-Khinchin representation gives (see\cite{JB})
\begin{equation}
\int_{0}^{\infty}e^{\ -z t}f_{X_s}(t) dt = e^{-s \psi_{X}(z)}
\end{equation}
where $\psi_{X}(\cdot)$ is the Laplace exponent of $X$ and is given by $\psi_{X}(z) = \int_{0}^{\infty} (1 -e^{-zu})\pi_X(du)$, where $\pi_X$ is the L\'evy measure of $X$. The following result can be found in \cite{JB}.
\begin{theorem}
\label{1sp}
The L\'evy density $w(x)$ and L\'evy measure $\pi_X(t,\infty)$ of the subordinator $X$ (with $\pi_X(t,\infty)=\int_{t}^{\infty} w(x) dx$) satisfy $\mathcal{L}(\pi_X(t,\infty)) =\frac{\psi_{X}(s)}{s}$, where $\psi_{X}(s)$ is the Laplace exponent of the subordinator $X$.
\end{theorem}

The following results are proved in \cite{AK}.
\begin{theorem}
\label{biggg}
Let $X= \{X_{t}\}_{t >0}$ be a subordinator with the probability density function $p(x,t)$. Suppose $p(x,t)$ admits continuous partial derivatives. Let $T_t= \inf\{\tau>0 : X_{\tau}\geq t\}$, for $t>0$, represents the first-exit time process of $X$. Denote the probability density function of $T_t$ by $h_t(\cdot)= h(\cdot,t)$. Then, 
\begin{equation}
\label{laplace}
\mathcal{L}(h(x,t)) = \frac{\psi_X (s) e^{-x\psi_X(s)}}{s},
\end{equation}
 where $\psi_X(\cdot)$ is the Laplace exponent of the subordinator $X$.
\end{theorem}
\begin{theorem}
\label{moment}
Denote the $q$-th moment of the first-exit time of the subordinator $X$ by $M_q(x,t)$. Then,
\begin{equation}
\mathcal{L}(M_q(x,t))=\frac{q\Gamma(1+q)}{s (\psi_X (s))^q }.
\end{equation}
\end{theorem}

\subsection{Gamma subordinators}
\label{sec42}

Let $X_t$ be a Gamma subordinator with L\'evy density given by $w_X(x) = \frac{\nu e^{-\alpha x}}{x}$, $x>0$, with $\nu, \alpha > 0$. In this case, the Laplace exponent of $X$ is given by $\psi_X (s) = \nu\ln\left(1+\frac{s}{\alpha} \right)$ (see \cite{Semere, NV}).

\begin{theorem}
\label{Ga1}
For $x \nu= n+1$, $n=0,1,2, \dots$, the probability density function of the first-exit time of $X$ is given by 
\begin{equation}
\label{hxt}
h(x,t) = \int_{0}^{t} \frac{e^{-u \alpha} \alpha^{xc}\left(\nu(-u)^{\nu x} {_2F_1\left(-\nu x,-\nu x,1-\nu x;1\right)}+\nu u^{\nu x}\right)}{(x\nu-1)!} \ du,
\end{equation}
where ${_2F_1\left(-\nu x,-\nu x,1-\nu x;1\right)}$ is the hypergeometric function. 
\begin{proof} 
By Theorem \ref{biggg}, the Laplace transform of probability density of the first-exit time of 
Gamma subordinator is given as 
\begin{equation*}
\mathcal{L}(h(x,t)) = \frac {\ln(1 + \frac{s}{\alpha}) ^ \nu} {s (1 + \frac{s}{\alpha})^ {x \nu} }= \frac{K(x,s)}{s},
\end{equation*}
where $K(x,s) = F(x,s)G(x,s)$, with $F(x,s) = \nu \ln(1+\frac{s}{\alpha})$,and $G(x,s) = \frac{1}{(1+\frac{s}{\alpha})^{x\nu}}$. Then $\mathcal{L}(h(x,t)) = \frac {K(x,s)}{s}$. Let the inverse Laplace transforms for $F(x,s)$ and $G(x,s)$ be $f(x,t)$ and $g(x,t)$, respectively.

Note that $\mathcal{L}^{-1}( \ln(1+s)) = -\frac{e^{-t}}{t}$, (see \cite{HK}). Using Lemma \ref{lem1}(1), we obtain, $$f(x,t) = -\frac{\nu e^{-t \alpha}}{t}.$$ 
For $x\nu= n+1$, where $n$ is a non-negative integer, $ \mathcal{L}^{-1}\left(-\frac{1}{(s+1)^{x\nu}}\right) =\frac{t^{x\nu-1}  e^{-t}}{(x\nu -1)!}$. Hence, by using Lemma \ref{lem1}(1), we obtain $g(x,t) =  \frac{\alpha^{x\nu} t^{x\nu -1} e^{-t \alpha}}{(x\nu-1)!}$. Consequently, by standard convolution procedure, we obtain
 \begin{align*}
k(x,t) & = \int_{0}^{t} f(x,\tau) g(x,t-\tau) d\tau  =  \int_{0}^{t} -  \frac{\nu e^{-\tau \alpha} \alpha^{x \nu}  (t-\tau)^{x \nu -1}  e^ {-\alpha(t-\tau)}}{(x \nu-1)!} d\tau \\
& = \frac{e^{-t\alpha} \alpha^{x \nu} [\nu(-t)^{\nu x} {_2F_1\left(-\nu x,-\nu x,1-\nu x;1\right)}+\nu t^{\nu x}]}{(x \nu-1)!}.
\end{align*} 
Hence, with the application of Lemma \ref{lem1}(3), we obtain \eqref{hxt}.
\end{proof}
\end{theorem}

The next result provides the first and the second order moment of the first-exit time of Gamma subordinator.

\begin{theorem}
The first order moment (mean) the first-exit time of Gamma subordinator $X_t$ is given by
\newcommand{\Int}{\int\limits}
\begin{equation}
\label{7ttt}
m(x,t) = \Int_{0}^{t}\Int_{0}^{\infty} \frac{\alpha  e^{-\lambda \alpha} ({\lambda \alpha})^{u-1}\,du \,d\lambda}{\nu \Gamma(u)}.
\end{equation} 
\end{theorem}
\begin{proof}
Using Theorem \ref{moment}, we obtain the Laplace transform of the $q$-th moment of the first-exit time of the Gamma subordinator as $\frac{q\Gamma(1+q)}{s (\psi_X (s))^q }$.
Consequently, the Laplace transform of the first order moment of the first-exit time of the Gamma subordinator is given by
\begin{equation*}
M(x,s)=\frac{\Gamma(2)}{s\nu\ln(1+\frac{s}{\alpha})}.
\end{equation*}
We observe that $\mathcal{L}^{-1}\left(\frac{\Gamma(2)}{\ln(s)}\right)=\int_{0}^{\infty}\frac{\Gamma(2)t^{u-1}}{\Gamma(u)} du$. 
Consequently, using Lemma \ref{lem1}(2), we obtain
\begin{equation*}
\mathcal{L}^{-1}\left(\frac{\Gamma(2)}{\nu \ln(1 + \frac{s}{\alpha})}\right)=\int_{0}^{\infty}\frac{ \alpha \Gamma(2) e^{-t\alpha}({t\alpha})^{u-1}}{ \nu \Gamma(u)}\,du.
\end{equation*}
Consequently, $\mathcal{L}^{-1}(M(x,s))$ can be computed using Lemma \ref{lem1}(3) to obtain \eqref{7ttt}.
\end{proof}

We conclude this subsection by considering the case when the subordinator $Z$, that appears in \eqref{shan41} and \eqref{se42}, is related to the Gamma subordinator in the BN-S model. As observed in Section \ref{sec2}, if the stationary distribution of $\sigma_t^2$ is given by gamma law $\Gamma(\nu, \alpha)$, then the L\'evy density of $Z_1$ is given by $w(x)= \nu \alpha e^{-\alpha x}$, $x>0$.

\begin{theorem}
\label{rama}
The probability density function of the first-exit time of a subordinator $Z$ with L\'evy density $w(x)= \nu\alpha e^{-\alpha x}$, is given by $$h(x,t)=\nu e^{-x \nu}I_0\left(2\sqrt{x \nu \alpha t}\right)e^{-\alpha t},$$ where $I_0(\cdot)$ is the modified Bessel function of the first kind.
\end{theorem}
\begin{proof}
For this case, the L\'evy measure of $Z$ is given by $\pi_Z(t,\infty) = \int_{t}^{\infty} \nu \alpha e^{-\alpha x} dx = \nu  e^{-\alpha t }$. Using Theorem \ref{1sp}, we obtain $\frac{\psi_{Z}(s)}{s} = \frac{\nu}{s+\alpha}$ . Consequently, $\psi_{Z}(s) = \frac{\nu s }{s+\alpha}$. The Laplace transform of the probability density function of the first-exit time of $Z$ is given by $H(x,s) = \frac{\nu e^{\frac{-x\nu s}{s+\alpha}}}{s+\alpha}$. Consequently, the probability density function of the first-exit time of $Z$ is given by $h(x,t)=\mathcal{L}^{-1} (H(x,s))$, where 
\begin{equation*}
H(x,s) = \frac{\nu e^{\frac{-x \nu s}{s+\alpha}}}{s+\alpha} = \frac{\nu e^{-x \nu (1-\frac{\alpha}{s+\alpha})}}{s+\alpha}= \frac{\nu e^{-x \nu}e^{\frac{x \nu \alpha}{s+\alpha}}}{s+\alpha}.
\end{equation*}
Using Lemma \ref{lem1}(1) and Lemma \ref{2sp}(2), we obtain $h(x,t) = \nu e^{-x \nu}I_0\left(2\sqrt{x \nu \alpha t}\right)e^{-\alpha t}$.
\end{proof}

\subsection{Inverse Gaussian subordinators}
\label{sec43}

The first-exit time of IG processes is described in \cite{AK}. In this subsection we consider the subordinator $Z$, that appears in \eqref{shan41} and \eqref{se42}, is related to the IG subordinator in the BN-S model. As observed in Section \ref{sec2}, if the stationary distribution of $\sigma^2_t$ is given by $\text{IG}(\delta_1, \gamma)$ law, then the L\'evy density of $Z_1$ is given by $w(x)= \frac{\delta_1}{2\sqrt{2 \pi}} x^{-\frac{3}{2}} (1+ \gamma^2 x)e^{-\frac{1}{2} \gamma^2 x}$, $x>0$, and $\delta_1, \gamma>0$. 

For the results in Subsections \ref{sec43} and \ref{sec44}, we define the \emph{convolution} of two functions $p(x,t)$ and $q(x,t)$ by $$p(x,t)*q(x,t)=\int_{0}^{t} p(x,\tau)q(x,t-\tau) d\tau.$$ Consequently, for three functions $p(x,t)$, $q(x,t)$, and $r(x,t)$,
\begin{align*}
(p(x,t)*q(x,t))*r(x,t) = \int_{0}^{t}\int_{0}^{u}p(x,\tau)q(x,u-\tau)r(x,t-u)\, d\tau \,du.
\end{align*}

\begin{theorem}
The probability density function of the first-exit time of a subordinator $Z$ with L\'evy density $w(x)=\frac{\delta_1}{2\sqrt{2 \pi}} x^{-\frac{3}{2}} (1+ \gamma^2 x)e^{-\frac{1}{2} \gamma^2 x}$, is given by 
$h(x,t)=(p(x,t)*q(x,t))*r(x,t)$, where
\begin{equation}
\label{ping}
p(x,t) =\frac{-\delta_1\gamma\erf(\frac{\gamma\sqrt{t}}{\sqrt{2}})} {2} + \frac{\delta_1 \gamma}{2} + \frac{\Gamma(\frac{-1}{2},\frac{\gamma^2 t}{2})\delta_1 \gamma}{4\sqrt{\pi}},
\end{equation}
\begin{equation}
\label{qing}
q(x,t) =\frac{e^{\frac{-x\gamma\delta_1}{2}}e^{\frac{-t\gamma^2}{2}}t^{\frac{-3}{2}}}{2\sqrt{\pi}} \int_{0}^{\infty} ue^{\frac{-u^2}{4t}}\left(I^{\prime}_0\left(2\sqrt{\frac{x\gamma^2\delta_1}{2\sqrt{2}}u}\right)\left(\sqrt{\frac{x\gamma^2\delta_1}{2\sqrt{2}u} }\right)+\delta(u)\right)\,du,
\end{equation}
where $\delta(\cdot)$ is the Dirac delta function, $I_0(\cdot)$ is the modified Bessel function of the first kind, and
\begin{equation}
\label{ring}
r(x,t)= \frac {x\delta_1\gamma^6 e^{\frac{-\gamma^2 t}{2}}e^{\frac{x\delta_1\gamma}{2}}e^{\frac{-\delta_1^2 x^2\gamma^4}{32 t}}} {8\sqrt{\pi}(2t)^{\frac{3}{2}}}.
 \end{equation}
\end{theorem}

\begin{proof}
We obtain the L\'evy measure for $Z$ as
\begin{align*}
\pi_Z(t,\infty) &= \int_{t}^{\infty} w(x) dx = \int_{t}^{\infty} \frac{\delta_1 x^\frac{-3}{2}e^\frac{-\gamma^2 x}{2} + \delta_1 (\gamma^2) x^\frac{-1}{2}e^\frac{-\gamma^2 x}{2}}{2\sqrt{2\pi}}\,dx\\
&= \frac{-\delta_1\gamma\erf(\frac{\gamma\sqrt{t}}{\sqrt{2}})} {2} + \frac{\delta_1 \gamma}{2} + \frac{\Gamma(\frac{-1}{2},\frac{\gamma^2 t}{•2})\delta_1 \gamma}{4\sqrt{\pi}}.
\end{align*}
Using Theorem \ref{1sp}, Lemma \ref{2sp}(4), and Lemma \ref{2sp}(5) we obtain,
\begin{equation}
\label{ein}
\mathcal{L}(\pi_Z(t,\infty))=\frac{\psi_{Z}(s)}{s}= \frac{\delta_1 \gamma}{2s} - \frac{\delta_1 \gamma^2}{2\sqrt{2}s(\sqrt{s+\frac{\gamma^2}{2}})} -\frac{\delta_1 \gamma [1-\sqrt{(1+\frac{2s}{\gamma^2})}]}{2s}.
\end{equation}
Consequently, by Theorem \ref{biggg}, we obtain that the Laplace transform of the probability density function of the first-exit time of $Z$ is given by
\begin{align*}
H(x,s) & = \left(\frac{\delta_1 \gamma}{2s} - \frac{\delta_1 \gamma^2}{2\sqrt{2}s(\sqrt{s+\frac{\gamma^2}{2}})} -\frac{\delta_1\gamma[1-\sqrt{(1+\frac{2s}{\gamma^2})}]}{2s}\right)e^{-x\left(\frac{\delta_1 \gamma}{2} - \frac{\delta_1 \gamma^2}{2\sqrt{2}\sqrt{s+\frac{\gamma^2}{2}}} -\frac{\delta_1\gamma[1-\sqrt{(1+\frac{2s}{\gamma^2})}]}{2}\right)} \\
&= P(x,s)Q(x,s)R(x,s),
\end{align*}
where 
\begin{equation*}
P(x,s) = \frac{\delta_1 \gamma}{2s} - \frac{\delta_1 \gamma^2}{2\sqrt{2}s \sqrt{s+\frac{\gamma^2}{2}}} -\frac{\delta_1\gamma[1-\sqrt{(1+\frac{2s}{\gamma^2})}]}{2s},
\end{equation*}
\begin{equation*}
Q(x,s) =\exp{\left(\frac{-x \gamma \delta_1}{2} + \frac{x\gamma^2 \delta_1}{{2\sqrt{2}\sqrt{s+\frac{\gamma^2}{2}}}}\right)},
\end{equation*}
and  
\begin{equation*}
R(x,s) = \exp{\left(\frac{x\delta_1\gamma}{2}- \frac{x\delta_1\gamma\sqrt{(1+\frac{2s}{\gamma^2})}}{2}\right)}.
\end{equation*}
We denote the inverse Laplace transforms  of $P(x,s)$, $Q(x,s)$, and $R(x,s)$ by $p(x,t)$, $q(x,t)$, and $r(x,t)$, respectively. 

We have $p(x,t) = \mathcal{L}^{-1} \left(\frac{\delta_1 \gamma}{2s} - \frac{\delta_1 \gamma^2}{2\sqrt{2}s(\sqrt{s+\frac{\gamma^2}{2}})} -\frac{\delta_1\gamma[1-\sqrt{(1+\frac{2s}{\gamma^2})}]}{2s}\right)$.
From this, comparing with \eqref{ein}, we note that $p(x,t)=\mathcal{L}^{-1}(\frac{\psi_Z(s)}{s})$. Hence $p(x,t)$ is given by \eqref{ping}.

Next, we compute $q(x,t)$ using Lemma \ref{2sp}(2), Lemma \ref{2sp}(1), and Lemma \ref{lem1}(4).
Lemma \ref{2sp}(2) gives $\mathcal{L}^{-1}\left(\frac{e^{\frac{a}{s}}}{s}\right) = I_0 (2\sqrt{at})$.

With $L(s) = \frac{e^{\frac{a}{s}}}{s}$, we find $l(t)= \mathcal{L}^{-1}(L(s))= I_0 (2\sqrt{at})$. We notice $I_0(0)=1$. Consequently, using Lemma \ref{lem1}(4), we have $\mathcal{L}^{-1} (sL(S) -l(0)) = l^{\prime}(t)$. Hence, we obtain, $\mathcal{L}^{-1} (e^{\frac{a}{s}}) - \mathcal{L}^{-1} (1) = I^{\prime}_0(2\sqrt{at})(\sqrt{\frac{a}{t}})$, and thus $\mathcal{L}^{-1}(e^{\frac{a}{s}}) = I^{\prime}_0(2\sqrt{at})(\sqrt{\frac{a}{t}}) + \delta(t)$, where $\delta(\cdot)$ is the Dirac delta-function.

Using Lemma \ref{2sp}(1), we obtain
$\mathcal{L}^{-1}(e^{as^{-1/2}}) = \frac{t^{\frac{-3}{2}}}{2\sqrt{\pi}}\int_{0}^{\infty} ue^{\frac{-u^2}{4t}}(I^{\prime}_0(2\sqrt{au})(\sqrt{\frac{a}{u}})+\delta(u)) du$. Therefore, using Lemma \ref{lem1}(1) we obtain
\begin{equation*}
q(x,t) = \frac{e^{\frac{-x\gamma\delta_1}{2}}e^{\frac{-t\gamma^2}{2}}t^{\frac{-3}{2}}}{2\sqrt{\pi}} \int_{0}^{\infty} ue^{\frac{-u^2}{4t}}\left(I^{\prime}_0\left(2\sqrt{\frac{x\gamma^2\delta_1}{2\sqrt{2}}u}\right)\left(\sqrt{\frac{x\gamma^2\delta_1}{2\sqrt{2}u} }\right)+\delta(u)\right)\,du.
\end{equation*}
Finally,
 \begin{equation*}
  r(x,t) = \mathcal{L}^{-1}(R(x,s)) =\mathcal{L}^{-1}\left(e^{\frac{x\delta_1\gamma}{2}}e^{-\frac{x\delta_1\gamma\sqrt{(1+\frac{2s}{\gamma^2})}}{2}}\right). 
 \end{equation*}
Using Lemma \ref{2sp}(3), we obtain $\mathcal{L}^{-1}(e^{\frac{x\delta_1\gamma}{2}}e^{-\frac{x\delta_1\gamma\sqrt{s}}{2}})$ = $\frac {e^{\frac{x\delta_1\gamma}{2}}(x\delta_1\gamma)e^{\frac{-\delta_1^2 x^2\gamma^2}{16 t}}} {4\sqrt{\pi}t^{\frac{3}{2}}}$. 
Consequently, using Lemma \ref{lem1}(1), we obtain \eqref{ring}.

Finally, if $h(x,t)$ is the probability density function of the first-exit time of $Z$, then $h(x,t)= \mathcal{L}^{-1}(H(x,s))= \mathcal{L}^{-1}(P(x,s)Q(x,s)R(x,s))=(p(x,t)*q(x,t))*r(x,t)$.
\end{proof}

\subsection{Positive tempered stable subordinators}
\label{sec44}

Let $X_t$ be a positive tempered stable (PTS) subordinator with L\'evy density given by $$u(x)= \beta k^{-2 \gamma} \frac{\gamma}{\Gamma(\gamma) \Gamma(1-\gamma)}x^{-\gamma-1}\exp\left(-\frac{1}{2}k^2 x\right), \quad x>0,$$
with $\beta>0$, $0<\gamma<1$, and  $k \geq 0$.

\begin{theorem}
\label{ptsnaive}
The probability density function of the first-exit time of $X$ is given by $h(x,t) = p(x,t)*q(x,t)$, where
\begin{equation}
\label{pptsxt}
p(x,t) = a\Gamma\left(-\gamma,\frac{k^2 t}{2}\right),
\end{equation}
where $a=\frac{\beta \gamma }{2^{\gamma} \Gamma(\gamma) \Gamma(1-\gamma)}$, and
\begin{align}
\label{qptsxt}
& q(x,t) =\nonumber \\
& e^{-xa\Gamma(-\gamma)}e^{{\frac{-k^{2}t}{2}}} \frac{(\frac{1}{\gamma})^{\frac{\gamma+2}{2\gamma}}}{(2\pi)^{\frac{\gamma+1}{2\gamma}}(-xa(\frac{2}{k^2})^\gamma \Gamma(-\gamma))^{\frac{1}{\gamma}}}\sum_{i,-i} \frac{1}{i} E\left(1,1+\gamma,\dots,2-\gamma ::\frac{((\frac{2}{k^2})^{\gamma}(-xa)\Gamma(-\gamma))^{\frac{1}{\gamma}}e^{i \pi}}{\gamma^{\frac{-1}{\gamma}}t}\right).
\end{align}
In \eqref{qptsxt}, $E(\cdot: \cdot : \cdot)$ is the MacRobert $E$-function, and $\sum_{i,-i}$ denotes that in expression following the summation sign, $i$ is to be replaced by $-i$ and two expressions are to be added.
\end{theorem}
\begin{proof}
We have
\begin{equation*}
\pi_X(t)=\int_{t}^{\infty} u(x) dx = \int_{t}^{\infty} \frac{\beta k^{-2\gamma}\gamma x^{-\gamma-1}e^{\frac{-k^2 x}{2}}}{\Gamma(\gamma)\Gamma(1-\gamma)} dx= \frac{\beta \gamma\Gamma(-\gamma,\frac{k^2 t}{2})}{\Gamma(\gamma)\Gamma(1-\gamma)2^{\gamma}}.
\end{equation*}
We compute $\mathcal{L}(\pi_X(t))$ using Theorem \ref{1sp} to obtain the Laplace exponent of density function of $X$ as 
\begin{equation*}
\psi_X(s) = \frac{\beta \gamma\Gamma(-\gamma)[1-(1+\frac{2s}{k^2})^{\gamma}]}{\Gamma(\gamma)\Gamma(1-\gamma)2^{\gamma}}.
\end{equation*}
Now using Theorem \ref{biggg}, we obtain the Laplace transform of the probability density function of the first-exit time of $X$ as
\begin{equation*}
H(x,s)=\mathcal{L}(h(x,t)) = \left(\frac{a\Gamma(-\gamma)[1-(1+\frac{2s}{k^2})^{\gamma}]}{s}\right)e^{-ax\Gamma(-\gamma)[1-(1+\frac{2s}{k^2})^{\gamma}]},
\end{equation*} 
where $a=\frac{\beta \gamma}{\Gamma(\gamma)\Gamma(1-\gamma)2^{\gamma}}$. To compute $h(x,t)$, the probability density function of the first-exit time of $X$, we find $$p(x,t)=\mathcal{L}^{-1}\left(\frac{a\Gamma(-\gamma)[1-(1+\frac{2s}{k^2})^{\gamma}]}{s}\right),$$ and $$q(x,t)=\mathcal{L}^{-1}\left(e^{-ax\Gamma(-\gamma)[1-(1+\frac{2s}{k^2})^{\gamma}]}\right),$$ and use the convolution result. By using Lemma \ref{2sp}(4), we obtain, the expression of $p(x,t)$ as \eqref{pptsxt}.

Next, compute $q(x,t)$. Denote $Q(x,s)=\exp\left(-ax\Gamma(-\gamma)[1-(1+\frac{2s}{k^2})^{\gamma}]\right)$. We observe
\begin{align*}
Q(x,s) &= e^{-xa\Gamma(-\gamma)}\exp\left( -\left({(-xa\Gamma(-\gamma))^{\frac{1}{\gamma}} + \frac{2s}{k^2}(-xa\Gamma(-\gamma))^{\frac{1}{\gamma}}}\right)^{\gamma}\right)\\
&=  e^{-xa\Gamma(-\gamma)}\exp\left( -\left({(-xa\Gamma(-\gamma))^{\frac{1}{\gamma}} + s\left(-xa\left(\frac{2}{k^2}\right)^{\gamma}\Gamma(-\gamma)\right)^{\frac{1}{\gamma}}}\right)^{\gamma}\right).
\end{align*}
Hence, by using Lemma \ref{2sp}(8) and Lemma \ref{lem1}(1), we obtain the expression of $q(x,t)$ as \eqref{qptsxt}.
\end{proof}

We conclude this subsection by considering a subordinator $Z$ related to the PTS subordinator in the BN-S model. If the stationary distribution of $\sigma_t^2$ is given by $\text{PTS}(\kappa, \delta, \gamma)$ law,  then that the L\'evy density of $Z_1$ is given by  
\begin{equation}
\label{pnew1}
w(x) = \frac{\beta k^{-2\gamma} \gamma x^{-\gamma-1} e^{\frac{-k^2 x}{2}}}{\Gamma(\gamma) \Gamma(1-\gamma)}\left(\gamma + \frac{k^2 x}{2}\right), \quad x>0, \quad \beta>0, 0<\gamma<1, k \geq 0.
\end{equation}
As in the previous sections, $Z$ is the subordinator that appears in \eqref{shan41} and \eqref{se42}.

\begin{theorem}
The probability density function of the first-exit time of a subordinator $Z$, with L\'evy density \eqref{pnew1}, is given by $h(x,t)=(p(x,t)*q(x,t))*r(x,t)$, where
\begin{equation}
\label{pmodified}
p(x,t) = a\left(\gamma\Gamma(-\gamma,\frac{k^2 t}{2}) + \Gamma(1-\gamma,\frac{k^2 t}{2})\right),
\end{equation}
where $a=\frac{\beta \gamma }{2^{\gamma} \Gamma(\gamma) \Gamma(1-\gamma)}$, and
 \begin{equation}
\label{qmodified}
q(x,t) = e^{-xa\gamma\Gamma(-\gamma)}e^{\frac{-k^{2}t}{2}} \frac{(\frac{1}{\gamma})^{\frac{\gamma+2}{2\gamma}}}{(2\pi)^{\frac{\gamma+1}{2\gamma}} (-xa\gamma(\frac{2}{k^2})^\gamma\Gamma(-\gamma))^\frac{1}{\gamma}}S_1,
 \end{equation}
where
\begin{equation*}
S_1=\sum_{i,-i} \frac{1}{i} E\left(1,1+\gamma,\dots,2-\gamma ::\frac{((\frac{2}{k^2})^{\gamma}(-xa\gamma)\Gamma(-\gamma))^{\frac{1}{\gamma}}e^{i \pi}}{\gamma^{\frac{-1}{\gamma}}t}\right),
\end{equation*}
and 
\begin{equation}
\label{rmodified}
r(x,t) = e^{-xa\Gamma(1-\gamma)}e^{-\frac{k^2t}{2}} \frac{(\frac{1}{\gamma-1})^{\frac{\gamma+1}{2\gamma-2}}}{(2\pi)^{\frac{\gamma}{2\gamma-2}}(-xa(\frac{2}{k^2})^{\gamma-1}\Gamma(1-\gamma))^\frac{1}{\gamma -1}}S_2,
\end{equation}
where
\begin{equation*}
S_2=\sum_{i,-i} \frac{1}{i} E\left(1,\gamma,\dots,3-\gamma ::\frac{((\frac{2}{\gamma^2})^{\gamma-1}(-xa)\Gamma(1-\gamma))^{\frac{1}{\gamma-1}}e^{i \pi}}{(\gamma-1)^{\frac{-1}{\gamma-1}}t}\right).
\end{equation*}
In the expressions of $S_1$ and $S_2$, $E(\cdot: \cdot : \cdot)$ is the MacRobert $E$-function, and $\sum_{i,-i}$ denotes that in expression following the summation sign, $i$ is to be replaced by $-i$ and two expressions are to be added.
\end{theorem}

\begin{proof}
We obtain $\pi_Z(t,\infty)$ as
\begin{align*}
\pi_Z(t,\infty) &= \int_{t}^{\infty} w(x) dx = \int_{t}^{\infty} \frac{\beta k^{-2\gamma} \gamma x^{-\gamma-1} e^{\frac{-k^2 x}{2}}}{\Gamma(\gamma) \Gamma(1-\gamma)}(\gamma + \frac{k^2 x}{2})\,dx \\
&= \frac{\beta \gamma}{2^{\gamma} \Gamma(\gamma) \Gamma(1-\gamma)}\left(\Gamma(-\gamma,\frac{k^{2}t}{2})\gamma + \Gamma(1-\gamma,\frac{k^{2} t}{2})\right).
\end{align*}
We use  Theorem \ref{1sp} to obtain
\begin{equation*}
\frac{\psi_{Z}(s)}{s}=\frac{\beta \gamma}{2^{\gamma} \Gamma(\gamma) \Gamma(1-\gamma)}\left(\frac{\gamma \Gamma(-\gamma)}{s}(1-(1 + \frac{2s}{k^2})^\gamma) +\frac{\Gamma(1-\gamma)}{s}(1-(1 + \frac{2s}{k^2})^{\gamma-1})\right). 
\end{equation*}
Consequently,
\begin{equation*}
\psi_{Z}(s) =a\left(\gamma \Gamma(-\gamma)(1-(1 + \frac{2s}{k^2})^\gamma) +\Gamma(1-\gamma)(1-(1 + \frac{2s}{k^2})^{\gamma-1})\right),
\end{equation*}
where $a=\frac{\beta \gamma }{2^{\gamma} \Gamma(\gamma) \Gamma(1-\gamma)}$. Using Theorem \ref{biggg}, we obtain the Laplace transform of the probability density function of the first-exit time of $Z$ as
\begin{align*}
H(x,s) = \mathcal{L}(h(x,t))= P(x,s) Q(x,s) R(x,s),
\end{align*}
where 
\begin{equation*}
P(x,s) = a\left(\frac{\gamma\Gamma(-\gamma)}{s} - \frac{\gamma\Gamma(-\gamma)}{s}(1+\frac{2s}{k^2})^\gamma + \frac{\Gamma(1-\gamma)}{s} -\frac{\Gamma(1-\gamma)}{s}(1+ \frac{2s}{k^2})^{\gamma-1}\right),
\end{equation*}
\begin{equation*}
Q(x,s) = \exp\left(-xa\left(\gamma \Gamma(-\gamma)(1-(1 + \frac{2s}{k^2})^\gamma)\right)\right),
\end{equation*}
and
\begin{equation*}
R(x,s) = \exp\left(-xa\left(\Gamma(1-\gamma)(1-(1 + \frac{2s}{k^2})^{\gamma-1}\right)\right).
\end{equation*}
We denote the inverse Laplace transform for $P(x,s)$, $Q(x,s)$, and $R(x,s)$, by $p(x,t)$, $q(x,t)$, and $r(x,t)$, respectively. Using Lemma \ref{2sp}(4), we obtain $\mathcal{L}^{-1}(\frac{\gamma\Gamma(-\gamma)}{s} (1-(1+\frac{2s}{k^2})^\gamma) =\gamma\Gamma(-\gamma,\frac{k^2 t}{2}) $.
Also, using Lemma \ref{2sp}(4), we obtain $\mathcal{L}^{-1} (\frac{\Gamma(1-\gamma)}{s} -\frac{\Gamma(1-\gamma)}{s}(1+ \frac{2s}{k^2})^{\gamma-1})=\Gamma(1-\gamma,\frac{k^2 t}{2})$. Hence,
we obtain $p(x,t)$ as given by \eqref{pmodified}.

Next, we observe that $Q(x,s)$ can be written as:
\begin{align*}
 Q(x,s)  &= e^{-xa\gamma\Gamma(-\gamma)}\exp\left( -\left({(-xa\gamma\Gamma(-\gamma))^{\frac{1}{\gamma}} + \frac{2s}{k^2}(-xa\gamma\Gamma(-\gamma))^{\frac{1}{\gamma}}}\right)^{\gamma}\right)\\
&=  e^{-xa\gamma\Gamma(-\gamma)}\exp\left(-\left({(-xa\gamma\Gamma(-\gamma))^{\frac{1}{\gamma}} + s\left(-xa\gamma\left(\frac{2}{k^2}\right)^{\gamma}\Gamma(-\gamma)\right)^{\frac{1}{\gamma}}}\right)^{\gamma}\right).
\end{align*}
Hence by using Lemma \ref{2sp}(8) and Lemma \ref{lem1}(1), we obtain \eqref{qmodified}. Finally, we observe that $R(x,s)$ can be written as
\begin{align*}
 R(x,s) &= e^{-xa\Gamma(1-\gamma)}\exp\left( -\left({(-xa\Gamma(1-\gamma))^{\frac{1}{\gamma-1}} + \frac{2s}{k^2}(-xa\Gamma(1-\gamma))^{\frac{1}{\gamma-1}}}\right)^{\gamma-1}\right)\\
 &= e^{-xa\Gamma(1-\gamma)}\exp\left( -\left({(-xa\Gamma(1-\gamma))^{\frac{1}{\gamma-1}} + s\left(-xa \left(\frac{2}{k^2}\right)^{\gamma-1}\Gamma(1-\gamma)\right)^{\frac{1}{\gamma-1}}}\right)^{\gamma-1}\right). 
\end{align*}
Hence by using Lemma \ref{2sp}(8) and Lemma \ref{lem1}(1), we obtain \eqref{rmodified}. Finally, by convolution theorem, we obtain the probability density function of the first-exit time of $Z$ as $h(x,t)=\mathcal{L}^{-1}(H(x,s))= (p(x,t)*q(x,t))*r(x,t)$.
\end{proof}

\section{Numerical results}
\label{sec5}

For this section, we use the S\&P 500 daily close price dataset for the period May 11, 2010 to May 8, 2020.
Table 1 summarizes some features of this empirical dataset. 

\begin{table}[h!]
\centering
\caption{Properties of the empirical dataset.}
  \begin{tabular}{ | l | c | }
    \hline
    & S\&P 500 daily close price  \\ \hline
    Mean & 2027.003   \\ \hline
    Median & 2036.709 \\ \hline
Maximum & 3386.149  \\     \hline
Minimum & 1022.580\\     \hline
  \end{tabular}
\end{table}

Figure 1 shows a line plot of the empirical dataset. The log-return process for the corresponding dataset is shown in Figure 2. Figure 3 and Figure 4 show the histograms of the S\&P 500 daily close price, and corresponding log-returns respectively.

For the empirical dataset we consider the log-return process $X_t$, with $X_0=0$. For the first-exit time process of the log-return, $\inf\{s>0: X_s \geq t\}$, we consider the associated first-exit time processes of the Brownian motion  $\inf\{s>0: W_{s}= t\}$, and the L\'evy subordinator  $\inf\{s>0: Z_{s} \geq t\}$. In the plots in Figure 5, we provide the histograms corresponding to the first-exit time of $X_t$ for the empirical dataset for various values of $t$. In the plots of Figure 6, we use Remark \ref{remar} to plot the probability density functions of $\inf\{s>0: W_{s}= t\}$, for $t=1,2,3,4$. Finally, we use Gamma-type subordinators described in Section \ref{sec42} with L\'evy density $w(x)= \nu\alpha e^{-\alpha x}$.
After finding appropriate parameter values, in the plots of Figure 7, we use Theorem \ref{rama} to plot the probability density functions of $\inf\{s>0: Z_{s} \geq t\}$, for $t=1,2,3,4$. From these figures, it is clear that for the time duration when there is no \emph{big fluctuation} of the empirical dataset, $\inf\{s>0: W_{s}= t\}$ plays the dominant role in determining the distribution of $\inf\{s>0: X_s \geq t\}$. However, for the time duration of \emph{big fluctuation} of the empirical dataset, $\inf\{s>0: Z_{s} \geq t\}$ plays the dominant role in determining the distribution of $\inf\{s>0: X_s \geq t\}$.
\begin{figure}[H]
\centering
\caption{S\&P 500 daily close price from May, 2010 -May, 2020.}
\includegraphics[scale=.6]{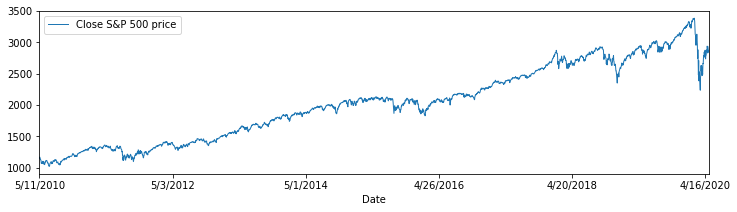}
\end{figure}
\begin{figure}[H]
\centering
\caption{S\&P 500 log-returns from May, 2010 -May, 2020.}
\includegraphics[scale=.6]{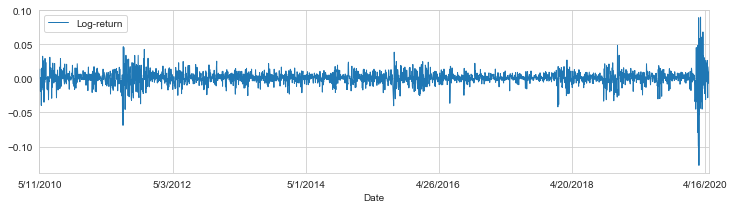}
\end{figure}

\begin{figure}[H]
\centering
\caption{Histogram for the S\&P 500 daily close price.}
\includegraphics[scale=.6]{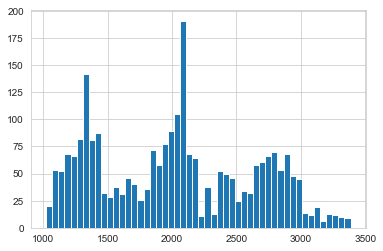}
\end{figure}

\begin{figure}[H]
\centering
\caption{Histogram for the log-return.}
\includegraphics[scale=.6]{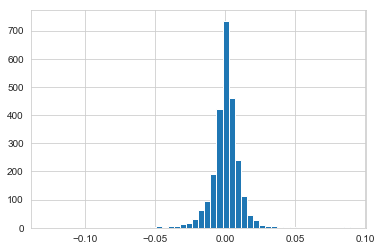}
\end{figure}

\begin{figure}[H]
\caption{Histograms corresponding to $\inf\{s>0: X_s \geq t\}$, for (left to right) $t=1,2,3,4$.}
\includegraphics[scale=.5]{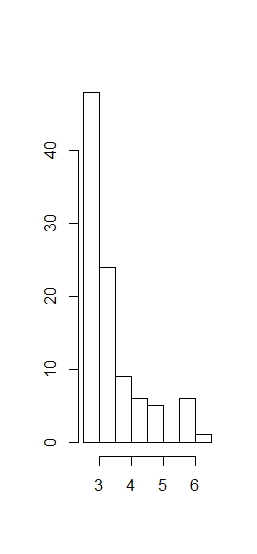}
\includegraphics[scale=.5]{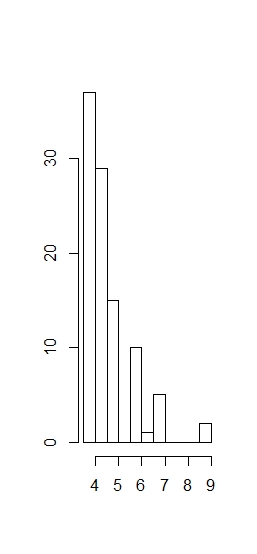}
\includegraphics[scale=.5]{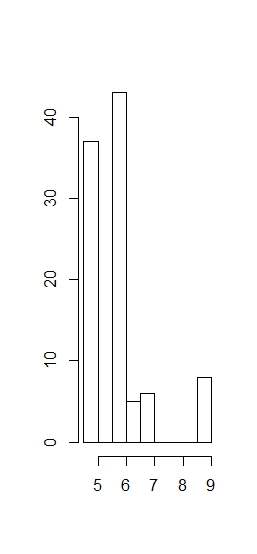}
\includegraphics[scale=.5]{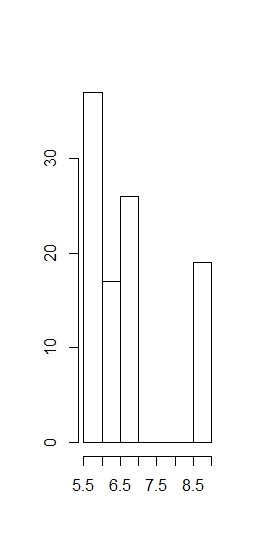}
\end{figure}

\begin{figure}[H]
\caption{Probability density functions of $\inf\{s>0: W_{s}= t\}$, for (left to right) $t=1,2,3,4$.}
\includegraphics[scale=.5]{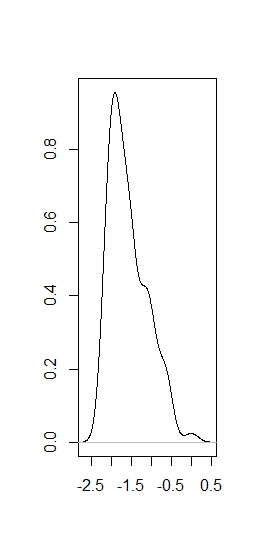}
\includegraphics[scale=.5]{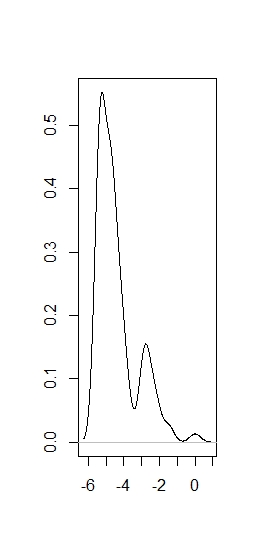}
\includegraphics[scale=.5]{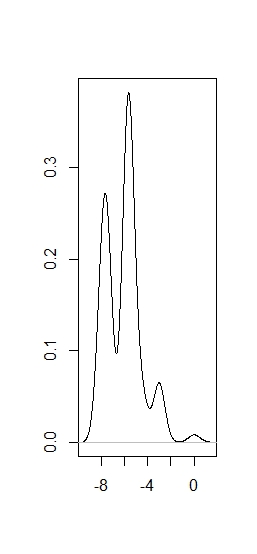}
\includegraphics[scale=.5]{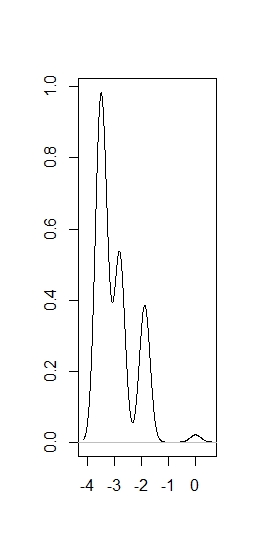}
\end{figure}

\begin{figure}[H]
\caption{Probability density functions of $\inf\{s>0: Z_{s} \geq t\}$, for (left to right) $t=1,2,3,4$.}
\includegraphics[scale=.5]{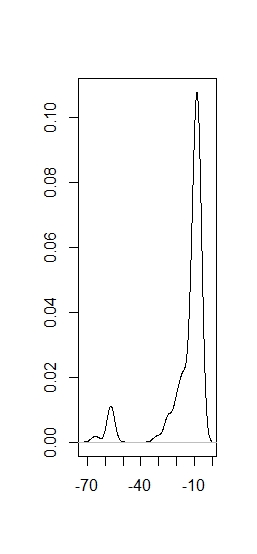}
\includegraphics[scale=.5]{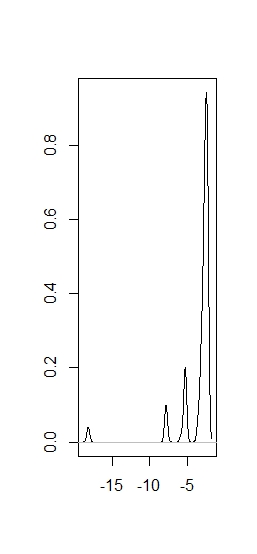}
\includegraphics[scale=.5]{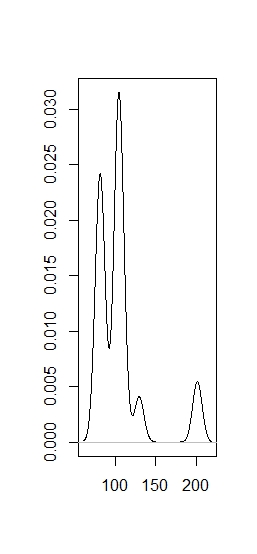}
\includegraphics[scale=.5]{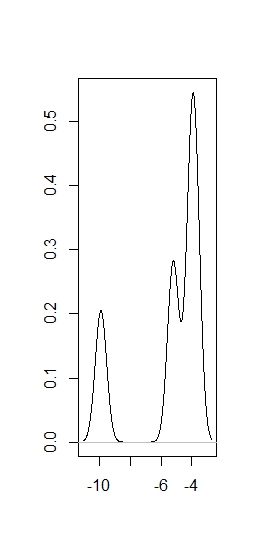}
\end{figure}


\section{Conclusion}
\label{sec6}

It is shown in this paper that an analytically tractable expression can be obtained for the probability density function of the first-exit time process of an approximate BN-S process. For the financial data, the density function of the first-exit time of the corresponding log-return process provides an important insight. In particular, such density function facilitates the understanding of a ``crash-like" future fluctuation of the market. In addition, this analysis has two-fold advantages. Firstly, based on the insight from the probability density function of the first-exit time process, the empirical data analysis for the future market is improved. Secondly, and more importantly, this provides a concrete way to improve existing stochastic models. For example, most of the existing financial models suffer from the lack of long-range dependence problem. An understanding of the density function of the first-exit time of stochastic models driven by a general L\'evy process can contribute positively to mitigate this issue. 

In the numerical results, we show various plots in support of the theoretical analysis provided in this paper. However, the analysis is dependent on the accurate estimation of model parameters for the empirical dataset. At present, we are implementing various machine learning based calibration techniques to improve the estimates of the parameter values for the empirical dataset. In effect, this may significantly improve the numerical results. These concepts, along with their connection to the first-exit time analysis, will be developed in a sequel of this paper.

\end{document}